\documentclass[letterpaper,10pt,twocolumn,twoside,journal]{IEEEtran}

\usepackage{times}
\usepackage{setspace}
\usepackage{url}
\spacing{1}
\usepackage[utf8]{inputenc}
\usepackage[T1]{fontenc}
\usepackage{graphicx}		
\usepackage{wrapfig}
\usepackage[format=plain,font=footnotesize,labelfont=bf,labelsep=period]{caption}
\usepackage{sidecap} 
\usepackage{subfig}
\usepackage[export]{adjustbox}
\usepackage[font=small]{caption}
\usepackage{float}

\usepackage{amsmath} 
\usepackage{amssymb}  
\usepackage{amsthm}
\usepackage{mathtools}
\usepackage{pdfsync}
\usepackage[normalem]{ulem}
\usepackage{paralist}	
\usepackage[space]{grffile} 
\usepackage{color}

\newtheorem{theorem}{Theorem}
\newtheorem{corollary}{Corollary}
\newtheorem{lemma}{Lemma}
\theoremstyle{definition}
\newtheorem{definition}{Definition}
\theoremstyle{remark}

\theoremstyle{definition}

\theoremstyle{definition}

\newcommand{\N}{\mathbb{N}}

\newcommand{\R}{\mathbb{R}}
\newcommand{\C}{\mathcal{C}}

\newcommand{\K}{\mathcal{K}}

\definecolor{darkblue}{RGB}{0,0,102}
\definecolor{lightblue}{RGB}{77,77,148}

\definecolor{gold}{RGB}{234, 170, 0}
\definecolor{metallic_gold}{RGB}{139, 111, 78}

\renewcommand{\cal}[1]{\mathcal{ #1 }}
\newcommand{\mb}[1]{\mathbf{ #1 }}

\newcommand{\derp}[2]{\frac{\partial #1 }{\partial #2 }}

\DeclareMathOperator*{\esssup}{ess\,sup}

\allowdisplaybreaks

\title{\LARGE \textbf{Safety-Critical Event Triggered Control via \\ Input-to-State Safe Barrier Functions}}
\author{Andrew J. Taylor$^{*}$, Pio Ong$^{*}$, Jorge Cort\'es, and Aaron D. Ames
\thanks{
$^*$Both authors contributed equally.}
\thanks{
A. Taylor is with the Department of Computing and Mathematical Sciences, California Institute of Technology, Pasadena, CA 91125, USA {\tt\small ajtaylor@caltech.edu}. P. Ong and J. Cort\'es are with the Department of Mechanical
    and Aerospace Engineering, University of California, San Diego, CA
    92093, USA, {\tt\small \{piong,cortes\}@eng.ucsd.edu}. A. Ames is with the Department of Mechanical \& Civil Engineering, California Institute of Technology, Pasadena, CA 91125, USA {\tt\small ames@caltech.edu}.}
}

\begin{document}

\maketitle

\begin{abstract}
The efficient utilization of available resources while simultaneously achieving control objectives is a primary motivation in the event-triggered control paradigm. In many modern control applications, one such objective is enforcing the safety of a system. The goal of this paper is to carry out this vision by combining event-triggered and safety-critical control design. We discuss how a direct transcription, in the context of safety, of event-triggered methods for stabilization may result in designs that are not implementable on real hardware due to the lack of a minimum interevent time. We provide a counterexample showing this phenomena and, building on the insight gained,  propose an event-triggered control approach via Input to State Safe Barrier Functions that achieves safety while ensuring that interevent times are uniformly lower bounded. We illustrate our results in simulation.
\end{abstract}

\vspace{-0.3cm}
\section{Introduction}\label{sec:intro}
Balancing control objectives such as stability and safety with the efficient use of available resources is critical in many modern control applications. Event-triggered control provides a framework that allows the prescription of, in a principled way, when certain resources (such as actuators, sensors, access to communication through a network or with neighboring agents, or even a human) should be utilized in order to guarantee the satisfaction of said control objectives efficiently, see e.g.,~\cite{tabuada2007event,heemels2012introduction,
borger2014event-separation, postoyan2014framework,CN-EG-JC:19-auto} and references therein.
The robustness properties of the original system enable the opportunistic, rather than continuous, use of the resources  without compromising the control objective.
In this paper we unify event-triggered control with \textit{Input-to-State Safe Barrier Functions} (ISSf-BF) to ensure a system remains safe while efficiently updating its actuation. 

Barrier Functions (BFs) \cite{xu2015robustness, ames2017control} have become increasingly popular as a tool for certifying the safety of a nonlinear control system \cite{ames2019control}. The impact on safety guarantees endowed by BFs introduced due to actuation errors was quantified in \cite{kolathaya2018input} through the notion of \textit{input-to-state safety} (ISSf) and Input-to-State Safe Barrier Functions. As the sample-and-hold implementation of controllers in the event-triggered context introduces errors in the actuation, ISSf provides a natural vehicle for designing trigger laws that guarantee a controller still achieves safety. The unification of event-triggered control with BFs is not a new idea \cite{ong2018event-triggered, yang2019self-triggered}, but is often done with the goal of improving stabilization rather than explicitly enforcing safety. Motivated by the work in \cite{tabuada2007event}, we use ISSf-BFs to quantify the impact of sample-and-hold actuation error on safety, and design a trigger law for updating actuation to enforce safety.  

One challenge in developing event-triggered control approaches for safety is ensuring that the time between events, or \textit{interevent times}, are lower bounded \cite{borger2014event-separation, postoyan2019inter-event, berneburg2019distributed}. Such bounds ensure that the resulting controller and trigger law can be practically implemented on systems that can not actuate infinitely quickly. The notion of Input to State Stability has been used to prove the existence of these bounds in the context of stabilization \cite{li2012stabilizing, tabuada2007event}. In contrast to the task of stabilization, in the context of safety the dynamics of the system, and thus the error dynamics, are not required to vanish as the quantity dictating the triggering of events vanishes. This can lead to events occurring in rapid succession. To ensure that interevent times are lower bounded, we leverage the notion of an \textit{input-to-state safe set} to ensure that the dynamics of the system evolve in a way that allows the quantity dictating the trigger to vanish. The end result is a concrete relationship between safety guarantees and minimum interevent times.

The main contribution of this paper is an event-triggered control paradigm for achieving safety of nonlinear control systems with uniformly lower bounded interevent times. Unlike stabilization, utilizing existing tools for safety in an event-triggered context may require additional assumptions or relaxations to yield a design that is implementable on real hardware. To this end, we define the \textit{strong ISSF-BF property}, and demonstrate how an arbitrarily small expansion of the safety set leads to uniformly lower bounded interevent times. To the best of our knowledge, \cite{yang2019self-triggered} is the only result that considers an event-triggered approach for achieving safety, but jointly considers it with stabilization and only addresses interevent time bounds corresponding to the latter. Our technical approach is the first to ensure safety and guarantee a minimum interevent time, and serves as a first step towards the development of a framework for the synthesis of safe controllers that optimize the usage of the system's resources.

\section{Event-Triggered Stability}
\label{sec:etstability}
In this section we discuss event-triggered control and review the problem of event-triggered stabilization following~\cite{tabuada2007event}. This review will motivate our approach for achieving event-triggered safety presented in Sections \ref{sec:lyapbar} and \ref{sec:etsafety}. 

Consider the nonlinear control system given by:
\begin{equation}
    \label{eqn:eom}
    \dot{\mb{x}} = \mb{f}(\mb{x},\mb{u}),
\end{equation}
where $\mb{x}\in\R^n$, $\mb{u}\in\R^m$, and $\mb{f}:\R^n\times\R^m\to\R^n$ is locally Lipschitz continuous in both arguments on $\R^n\times\R^m$. We further assume that $\mb{f}(\mb{0},\mb{u}_0)=\mb{0}$ for some $\mb{u}_0\in\R^m$. Under the choice of a Lipschitz continuous state-feedback controller $\mb{k}:\R^n\to\R^m$, with $\mb{k}(\mb{0})=\mb{u}_0$, the closed-loop system dynamics are given by:
\begin{equation}
    \label{eqn:cloop}
    \dot{\mb{x}} = \mb{f}(\mb{x},\mb{k}(\mb{x})).
\end{equation}
The assumption on local Lipschitz continuity of $\mb{f}$ and $\mb{k}$ implies that for any initial condition $\mb{x}_0 := \mb{x}(0) \in \R^n$, there exists a maximum time interval $I(\mb{x}_0) = [0, t_{\textrm{max}})$ such that $\mb{x}(t)$ is the unique solution to \eqref{eqn:cloop} on $I(\mb{x}_0)$. In the case $\mb{f}(\cdot,\mb{k}(\cdot))$ is forward complete, $t_{\textrm{max}}=\infty$. 

In an event-triggered context, the implementation of the feedback control law $\mb{k}$ is done by sampling the state at sequential time instances, $t_0, t_1, t_2, \ldots$, and evaluating the controller on the corresponding states $\mb{x}(t_0), \mb{x}(t_1), \mb{x}(t_2), \ldots$ Between measurements the control input is held constant:
\begin{equation}
\label{eqn:samplehold}
    \mb{u}(t) = \mb{k}(\mb{x}(t_i)) \quad \forall t\in[t_i, t_{i+1}).
\end{equation}
The time instances at which the controller is updated are determined by a state-dependent execution rule or trigger law. We define the measurement error as:
\begin{equation}
    \mb{e}(t) = \mb{x}(t_i)-\mb{x}(t) \qquad \forall t\in[t_i, t_{i+1}),
\end{equation}
noting that this in conjunction with \eqref{eqn:samplehold} implies:
%
%
\begin{equation}\label{eqn:cloopZOH}
    \dot{\mb{x}}(t) = -\dot{\mb{e}}(t)= \mb{f}(\mb{x}(t),\mb{k}(\mb{x}(t_i))) \qquad \forall t\in[t_i, t_{i+1}),
\end{equation}
The closed-loop dynamics \eqref{eqn:cloopZOH} can alternatively be written as:
\begin{equation}
    \label{eqn:clooperr}
    \dot{\mb{x}} = \mb{f}(\mb{x},\mb{k}(\mb{x}+\mb{e})).
\end{equation}
where $\mb{e}\in\R^n$ is the measurement error.

Event-triggered stabilization relies on the robustness to disturbances of the original dynamics, formalized through Input-to-State Stable Lyapunov Functions (ISS-LF)~\cite{sontag1995characterizations2, sontag2008input}.

\noindent \textit{Notation:}  Throughout the paper we make use of the following basic definitions. A continuous function $\alpha:[0,a)\to\R_+$, with $a>0$, is \textit{class $\cal{K}$} ($\alpha\in\cal{K}$) if $\alpha(0)=0$ and $\alpha$ is strictly monotonically increasing. If $a=\infty$ and $\lim_{r\to\infty}\alpha(r)=\infty$, then $\alpha$ is \textit{class $\cal{K}_\infty$} ($\alpha\in\cal{K}_\infty$). A continuous function $\alpha:(-b,a)\to\R$, with $a,b>0$, is \textit{extended class $\cal{K}$} ($\alpha\in\cal{K}_e$) if $\alpha(0)=0$ and $\alpha$ is strictly monotonically increasing. If $a,b=\infty$, $\lim_{r\to\infty}\alpha(r)=\infty$, and $\lim_{r\to-\infty}\alpha(r)=-\infty$, then $\alpha$ is \textit{extended class $\cal{K}_\infty$} ($\alpha\in\cal{K}_{\infty,e}$)

\begin{definition}[\textit{ISS Lyapunov Function (ISS-LF)}]
A continuously differentiable function $V:\R^n\to\R_+$ is an \textit{Input to State Stable Lyapunov Function} (ISS-LF) for \eqref{eqn:clooperr}, with respect to measurement errors $\mb{e}$, if there exists $\alpha_1,\alpha_2,\alpha_3\in\cal{K}_\infty$ and $\gamma\in\cal{K}_\infty$ such that for all $\mb{x},\mb{e}\in\R^n$:
\begin{subequations}
\begin{align}
    \label{eqn:Vbounds}\alpha_1(\Vert\mb{x}\Vert_2) & \leq V(\mb{x})  \leq \alpha_2(\Vert\mb{x}\Vert_2), \\ \label{eqn:Vdotbound} \derp{V}{\mb{x}}(\mb{x})\mb{f}(\mb{x},\mb{k}(\mb{x}+\mb{e})) & \leq -\alpha_3(\Vert\mb{x}\Vert_2)+\gamma(\Vert\mb{e}\Vert_2). 
\end{align}
\end{subequations}
\end{definition}

As in \cite{tabuada2007event}, if we define the trigger law to enforce:
\begin{equation}
\label{eqn:triggerineq}
    \gamma(\Vert \mb{e}(t) \Vert_2) \leq \sigma \alpha_3(\Vert\mb{x}(t)\Vert_2) \qquad 0<\sigma<1,
\end{equation}
the ISS-LF condition \eqref{eqn:Vdotbound} leads to:
\begin{equation}
\label{eqn:eventVdot}
    \derp{V}{\mb{x}}(\mb{x})\mb{f}(\mb{x},\mb{k}(\mb{x}+\mb{e})) \leq (\sigma-1)\alpha_3(\Vert\mb{x}\Vert_2),
\end{equation}
implying asymptotic stability of \eqref{eqn:clooperr} to $\mb{x}^\star = \mb{0}$. The inequality in \eqref{eqn:triggerineq} can be enforced by defining the trigger law as:
\begin{equation}
\label{eqn:stabletrigger}
   t_{i+1} = \min\left\{t\geq t_i ~|~ \gamma(\Vert\mb{e}(t)\Vert_2) = \sigma\alpha_3(\Vert\mb{x}(t)\Vert_2)\right\}.
\end{equation}
As is typical in event-triggered control formulations, it is critical to show that such a trigger law does not lead to the control being updated at arbitrarily close time instances \cite{postoyan2014framework}, or that the \textit{interevent times} $\{t_{i+1}-t_i\}_{i\in\N}$ are lower bounded by a positive constant $\tau\in\R$, $\tau>0$, referred to as the minimum interevent time (MIET).
%
%
This differs slightly from preventing the stronger notion of Zeno behavior \cite{borger2014event-separation},
%
%
in which the series of interevent times converges (implying the lack of a MIET). The results of \cite{tabuada2007event} ensure that a MIET exists under the trigger~\eqref{eqn:eventVdot}.

\section{Input-to-State Safety}
\label{sec:BFs}
In this section we provide background information on Barrier Functions (BFs) and Input to State Safe Barrier Functions (ISSf-BFs) that will be used to construct an event-triggered control paradigm that ensures safety.

Consider a set $\C\subseteq \R^n$ defined as the 0-superlevel set of a continuously differentiable function $h:\R^n \to \R$, yielding:
\begin{subequations}
\begin{align}
    \C &\triangleq \left\{\mb{x} \in \R^n : h(\mb{x}) \geq 0\right\}, \label{eqn:safeset}\\
    \partial\C &\triangleq \{\mb{x} \in \R^n : h(\mb{x}) = 0\},\label{eqn:safesetboundary}\\
    \textrm{Int}(\C) &\triangleq \{\mb{x} \in \R^n : h(\mb{x}) > 0\}.\label{eqn:safetsetinterior}
\end{align}
\end{subequations}
We assume that $\C$ is nonempty and has no isolated points, that is, $\textrm{Int}(\C) \not = \emptyset \textrm{ and } \overline{\textrm{Int}(\C)} = \C$. We refer to $\cal{C}$ as the \textit{safe set}. This construction motivates the following definitions: 

\begin{definition}[\textit{Forward Invariant \& Safety}]
A set $\C$ is \textit{forward invariant} if for every $\mb{x}_0\in\C$, the solution $\mb{x}(t)$ to \eqref{eqn:cloop} satisfies $\mb{x}(t) \in \C$ for all $t \in I(\mb{x}_0)$. The system \eqref{eqn:cloop} is \textit{safe} on $\C$ if the set is forward invariant.
\end{definition}

Verifying that the system \eqref{eqn:cloop} is safe on a set $\C$ can be done via a Barrier Function:

\begin{definition}[\textit{Barrier Function (BF)}]
A continuously differentiable function $h:\R^n\to\R$ is a \textit{Barrier Function} (BF) for \eqref{eqn:cloop} on a set $\cal{C}\subset\R^n$ defined as in \eqref{eqn:safeset}-\eqref{eqn:safetsetinterior}, if there exists $\alpha\in\K_{\infty,e}$ such that for all $\mb{x}\in\R^n$:
\begin{equation}
\label{eqn:barrier}
    \derp{h}{\mb{x}}(\mb{x})\mb{f}(\mb{x},\mb{k}(\mb{x})) \geq -\alpha(h(\mb{x})),
\end{equation}
\end{definition}
As shown in \cite{xu2015robustness}, the existence of a barrier function for \eqref{eqn:cloop} on a set $\C$ is sufficient to prove the safety and asymptotic stability of~$\C$.
To consider the impact of measurement errors on safety, we consider the notion of input-to-state safety~\cite{kolathaya2018input}.

\begin{definition}[\textit{Input to State Safety (ISSf)}]
Let the signal $\mb{e}:\R_+\to\R^n$ be essentially bounded and define $\Vert\mb{e}\Vert_\infty = \esssup_{t\geq0}\Vert\mb{e}(t)\Vert_2$. The closed-loop system \eqref{eqn:clooperr} is \textit{input-to-state safe} (ISSf) on $\C$ with respect to measurement errors $\mb{e}$ if there exists $\gamma\in\cal{K}_\infty$ and a set $\C_\mb{e}\supset\C$ defined as:
\begin{subequations}
\begin{align}
    \C_{\mb{e}} &\triangleq \left\{\mb{x} \in \R^n : h(\mb{x})+\gamma(\Vert\mb{e}\Vert_\infty) \geq 0\right\}, \label{eqn:safeseterr}\\
    \partial\C_{\mb{e}} &\triangleq \{\mb{x} \in \R^n : h(\mb{x})  +\gamma(\Vert\mb{e}\Vert_\infty) = 0\},\label{eqn:safesetboundaryerr}\\
    \textrm{Int}(\C_{\mb{e}}) &\triangleq \{\mb{x} \in \R^n : h(\mb{x})+\gamma(\Vert\mb{e}\Vert_\infty) > 0\},\label{eqn:safetsetinteriorerr}
\end{align}
\end{subequations}
such that \eqref{eqn:clooperr} is safe on $\C_\mb{e}$.
\end{definition}

We refer to $\C$ as an \textit{input-to-state safe} set (ISSf set) if such a set $\C_{\mb{e}}$ exists. This definition implies that though the set $\C$ may not be safe, a larger set $\C_\mb{e}$, depending on $\mb{e}$, is safe.
%
%
If $\mb{e}\equiv\mb{0}$, we recover that the set $\C$ is safe. This motivates the following definition of Input-to-State Safe Barrier Functions:

\begin{definition}[\textit{Input to State Safe Barrier Function (ISSf-BF)}]
A continuously differentiable function $h:\R^n\to\R$ is an \textit{Input-to-State Safe Barrier Function} (ISSf-BF) for \eqref{eqn:clooperr} on a set $\cal{C}\subset\R^n$ defined as in \eqref{eqn:safeset}-\eqref{eqn:safetsetinterior}, if there exists $\alpha\in\K_{\infty,e}$ and $\iota\in\cal{K}_\infty$ such that for all $\mb{x},\mb{e}\in\R^n$:
\begin{equation}
\label{eqn:ISSf-BF}
    \derp{h}{\mb{x}}(\mb{x})\mb{f}(\mb{x},\mb{k}(\mb{x}+\mb{e})) \geq -\alpha(h(\mb{x}))-\iota(\Vert\mb{e}\Vert_2),
\end{equation}
\end{definition}
As shown in \cite{kolathaya2018input}, the existence of an ISSf-BF for \eqref{eqn:clooperr} on $\C$ implies $\C$ is an ISSf set, implying safety and asymptotic stability to the set $\C_{\mb{e}}$.
%
%

\section{Towards Resource-Aware Safety: from Lyapunov to Barriers}\label{sec:lyapbar}
%

%
To more efficiently utilize actuation resources when implementing safe controllers, we seek to unify the preceding concepts of event-triggered control and input-to-state safety. In this section we discuss challenges that arise in directly transferring ideas from event-triggered stabilization to safety. Given the similarity of the ISS-LF constraint \eqref{eqn:Vdotbound} and the ISSf-BF constraint \eqref{eqn:ISSf-BF}, it is natural to propose a trigger law that enforces: 
\begin{equation}
\label{eqn:safetriggerineq}
    \iota(\Vert \mb{e}(t) \Vert_2) \leq \sigma \alpha(h(\mb{x}(t))) \qquad 0<\sigma,
\end{equation}
implying:
\begin{equation}
\label{eqn:hdotbound}
    \derp{h}{\mb{x}}(\mb{x})\mb{f}(\mb{x},\mb{k}(\mb{x}+\mb{e})) \geq -(1+\sigma)\alpha(h(\mb{x})).
\end{equation}
This can be interpreted as allowing the system to more quickly approach the boundary of the safe set at the expense of actuation resources. It is important to note that inside the set $\C$ it is possible to satisfy \eqref{eqn:safetriggerineq} and thus enforce safety of $\C$, but the inequality is impossible to satisfy outside of the set $\C$ as $\alpha(h(\mb{x}))<0$ if $\mb{x}\notin\C$. This type of behavior does not arise in the context of event-triggered stabilization, where convergence is to a point. One method to solve this issue is to instead define the trigger law:
\begin{equation}
\label{eqn:safetriggerineqsigned}
    \iota(\Vert \mb{e}(t) \Vert_2) \leq \sigma \vert\alpha(h(\mb{x}(t)))\vert \qquad 0<\sigma<1,
\end{equation}
which enforces \eqref{eqn:hdotbound} if $\mb{x}\in\C$ and enforces:
\begin{equation}
\label{eqn:hdotboundoutside}
    \derp{h}{\mb{x}}(\mb{x})\mb{f}(\mb{x},\mb{k}(\mb{x}+\mb{e})) \geq -(1-\sigma)\alpha(h(\mb{x})),
\end{equation}
if $\mb{x}\notin\C$. In this formulation, the system is not only allowed to more quickly approach the boundary, but is also not required to converge to the set as quickly when outside of the set. This is a generalization of event-triggered stabilization to a set.


Even with this solution, it is not guaranteed that this trigger law will have a MIET. Although ruling out Zeno behavior is not required to guarantee safety, unlike stabilization, it is important to have a MIET in term of implementation of the controller (cf. \cite{borger2014event-separation}). The key difference between stability and safety leading to the failure of a MIET to exist for a safe event-triggered controller lies in how the system dynamics must behave close to an equilibrium point compared to how they can behave close to the boundary of the safe set. In stabilization, continuity of the dynamics requires the dynamics to vanish as the equilibrium is approached, leading to the error dynamics in \eqref{eqn:cloopZOH} vanishing. In safety, the dynamics close to the boundary of the safe set need not vanish as the boundary is approached, such that the error dynamics in \eqref{eqn:cloopZOH} need not vanish. We provide the following counterexample to illustrate how this difference can lead to an MIET failing to exist for a safe controller.

\subsection{Counterexample}
%
%
Consider the following system:
\begin{equation}
\label{eqn:cex}
    \frac{\mathrm{d}}{\mathrm{d}t}\begin{bmatrix}x_1 \\ x_2 \end{bmatrix} = \begin{bmatrix}x_2 \\ -x_1 \end{bmatrix} +\begin{bmatrix}x_1 \\ x_2 \end{bmatrix}u.
\end{equation}
for which we wish to ensure the safety of the set $\C$, given by the 0-superlevel set of the continuously differentiable function $h(\mb{x})=1-x_1^2-x_2^2=1-\Vert\mb{x}\Vert_2^2$. The time derivative of this function along solutions to \eqref{eqn:cex} is given by $\dot{h}(\mb{x},u) = -2(x_1^2+x_2^2)u$, for which the state-feedback controller $k(\mb{x}) = \frac{1}{2}(1-x_1^2-x_2^2)$ yields $\dot{h}(\mb{x}) = -(x_1^2+x_2^2)h(\mb{x})\geq -h(\mb{x})$, which implies $h$ is a valid BF for the set $\C$.

In an event triggered context, the closed-loop dynamics of the system are then given by:
\begin{equation}
\label{eqn:cexcloop}
\dot{\mb{x}}(t) = \begin{bmatrix}k(\mb{x}(t_i)) & 1 \\ -1 & k(\mb{x}(t_i)) \end{bmatrix}\mb{x}(t),
\end{equation}
for each time $t\in [t_i,t_{i+1})$. This leads to the time derivative of $h$ along solutions to \eqref{eqn:cexcloop} being given by:
\begin{equation*}
\dot{h}(\mb{x}(t),\mb{e}(t)) = -\Vert\mb{x}(t)\Vert_2^2h(\mb{x}(t_i)) =   -\Vert\mb{x}(t)\Vert_2^2h(\mb{x}(t)+\mb{e}(t))
\end{equation*}
for each time $t\in [t_i,t_{i+1})$, where $\mb{e}(t)=\mb{x}(t_i)-\mb{x}(t)$.
To see that the BF $h$ is in fact an ISSf-BF, we note its time derivative can be bounded as follows (omitting the explicit dependence on time):
\begin{align*}
    \dot{h}(\mb{x},\mb{e}) &= -\Vert\mb{x}\Vert_2^2 h(\mb{x}+\mb{e})\\
    &=-\Vert\mb{x}\Vert_2^2 (1-\Vert\mb{x}\Vert_2^2-2\mb{x}^\top \mb{e}-\Vert\mb{e}\Vert_2^2)\\
    &\geq -\Vert\mb{x}\Vert_2^2 (1-\Vert\mb{x}\Vert_2^2)-2\Vert\mb{x}\Vert_2^3\Vert\mb{e}\Vert_2\\
    &\geq -(1-\Vert\mb{x}\Vert_2^2)-2\Vert\mb{x}\Vert_2^3\Vert\mb{e}\Vert_2\\
    &\geq -h(\mb{x})-2r^3\Vert\mb{e}\Vert_2,
\end{align*}
for $\Vert\mb{x}\Vert_2\leq r$. 
Given that $h$ is an ISSf-BF on some domain containing the unit circle (choose $r>1$), the trigger law enforcing \eqref{eqn:safetriggerineqsigned} is given by:
\begin{equation}
\label{eqn:cextrigger}
t_{i+1} = \min\{t\geq t_i ~|~ 2r^3\Vert\mb{e}(t)\Vert_2=\sigma |h(\mb{x}(t))|\},
\end{equation}
with $0<\sigma<1$. This will guarantee that $\C$ is safe as:
\begin{equation*}
\dot h(\mb{x},\mb{e}) \geq \begin{cases} -(1+\sigma)h(\mb{x}), & \Vert\mb{x}\Vert_2\leq 1 ,
\\ 
-(1-\sigma)h(\mb{x}), & 1< \Vert\mb{x}\Vert_2\leq r .
\end{cases}
\end{equation*}
%
%
\begin{lemma}[MIET does not exist]\label{le:no-MIET}
The system \eqref{eqn:cexcloop} with the trigger law defined as in \eqref{eqn:cextrigger} does not possess an MIET.
\end{lemma}
\begin{proof}
To show that the interevent times $\{t_{i+1}-t_i\}_{i\in\N}$ are not lower bounded, we will proceed via contradiction. In particular, let us assume that there exists $\tau>0$ such that $t_{i+1}-t_i\geq\tau$ for all $i\in\N$. If the state $\mb{x}_i=\mb{x}(t_i)$ at event time $t_i$ is used as an initial condition, the solution to~\eqref{eqn:cexcloop} is: 
\begin{align*}
\mb{x}(t)& = \exp{\left(\frac{h(\mb{x}_i)\Delta t_i}{2}\right)}\begin{bmatrix}\cos{\Delta t_i} &\sin{\Delta t_i}\\-\sin{\Delta t_i}& \cos{\Delta t_i}\end{bmatrix}\mb{x}_i, \nonumber \\ & = \mb{M}_i(\Delta t_i)\mb{x}_i,
\end{align*}
for $t\in[t_i,t_{i+1})$ with $\Delta t_i = t-t_i$. Denoting $\omega_i=h(\mb{x}_i)\Delta t_i$, we see that the norm of the error is lower bounded by a function monotonically increasing in time:
\begin{align*}
    \Vert\mb{e}(t)\Vert_2 &=  \Vert(\mb{I}-\mb{M}_i(\Delta t_i))\mb{x}_i\Vert_2, \\
        &=  \sqrt{\left(\exp{(\omega_i)}-2\exp{\left(\frac{\omega_i}{2}\right)}\cos(\Delta t_i)+1\right)}\Vert\mb{x}_i\Vert_2,\\
        &\geq  \left\lvert \exp{\left(\frac{\omega_i}{2}\right)}-1\right\rvert\Vert\mb{x}_i\Vert_2.
\end{align*}
This lower bound on the error grows unbounded in time. This implies that no matter the state in $\textrm{Int}(\C)$ that an event occurs, another event must occur at some time in the future (or the bound in \eqref{eqn:cextrigger} will be violated as $h$ is upper bounded on the safe set $\C$). Thus, for all $T>0$, there exists an event time $t_i>T$.

Next, we will show that $\lim_{t\to\infty}h(\mb{x}(t))=0$. The state solution yields:
\begin{equation*}
h(\mb{x}(t)) = 1- \Vert\mb{x}(t)\Vert_2^2 = 1-  \exp{(h(\mb{x}_i) \Delta t_i)}\Vert\mb{x}_i\Vert_2^2,
\end{equation*}
with time derivative:
\begin{equation*}
    \dot{h}(\mb{x}(t)) = -h(\mb{x}_i)\exp{(h(\mb{x}_i) \Delta t_i)}\Vert\mb{x}_i\Vert_2^2,
\end{equation*}
for $t\in[t_i,t_i+1)$. Within the safe set we have that $\dot{h}(\mb{x}(t))\leq 0$, such that $h(\mb{x}(t))$ is monotonically decreasing in time. The safety of $\C$ implies $h(\mb{x}(t))$ is lower bounded by $0$, and thus we can conclude that $\lim_{t\to\infty}h(\mb{x}(t))$ exists. Assume that this limit is some value $0<c<1$. Thus for any $\delta > 0$, there exists $T>0$ such that for $t>T$, $h(\mb{x}(t)) < c+\delta$. Since there are an infinite number of events, we deduce there exists $t_i>T$ such that $h(\mb{x}(t_i))<c+\delta$. As $h(\mb{x}(t))$ is monotonically decreasing, it also follows $h(\mb{x}(t))\geq c$ for all $t$. This implies:
\begin{equation*}
    \dot{h}(\mb{x}(t)) \leq -c\exp{(c \Delta t_i)}(1-(c+\delta)) \leq -c+c^2+c\delta < 0.
\end{equation*}
for $t\geq t_i$ where $\delta$ is chosen small enough to enforce the strict inequality with $0$. Thus between two events we have:
\begin{equation*}
    h(\mb{x}(t_{i+1})) \leq h(\mb{x}(t_i)) + \tau(-c+c^2+c\delta),
\end{equation*}
where $\tau$ is the assumed MIET. 
 Choosing $\delta < \tau(c-c^2)/(1+\tau c)$ implies $h(\mb{x}(t_{i+1})) < c$, contradicting the assumption that $c\neq 0$ (and maintaining the assumption on the existence of~$\tau$).


To complete the proof, note $\mb{e}(t_i)=\mb{0}$ and take the second-order (one-sided) Taylor expansion of $\Vert\mb{e}(t)\Vert_2^2$ at $t= t_i$:
\begin{align*}
    \Vert\mb{e}(t)\Vert_2^2 & = \left(\dot{\mb{e}}(t_i)^\top\dot{\mb{e}}(t_i)\right)(t-t_i)^2+\cal{O}((t-t_i)^3), \nonumber \\
    & = (1+k(\mb{x}(t_i)))\Vert\mb{x}(t_i)\Vert_2^2(t-t_i)^2+\cal{O}((t-t_i)^3), \nonumber \\ 
    & \geq \Vert\mb{x}(t_i)\Vert_2^2(t-t_i)^2 -c_3(t-t_i)^3,
\end{align*}
with $c_3>0$. The first term in the inequality follows from $k(\mb{x}(t_i))\geq 0$ for $\mb{x}(t_i)\in\C$. The second term follows if we view $\frac{\mathrm{d}^3}{\mathrm{d}t^3}\Vert\mb{e}(t)\Vert_2^2$ as a continuous function of the state, which remains within the compact set $\C$. Then $\frac{\mathrm{d}^3}{\mathrm{d}t^3}\Vert\mb{e}(t)\Vert_2^2$ will be bounded for all time as $\C$ is forward invariant. We can use this bound in conjunction with Lagrange's Remainder Formula~\cite{abbott2001understanding} to assert the existence of $c_3$.

At trigger time $t_i$, let $h(\mb{x}(t_i)) = \epsilon_i$ with $\epsilon_i >0$ arbitrarily small due to the existence of infinite triggers and convergence of $h$. This implies $\Vert\mb{x}(t_i)\Vert_2^2=1-\epsilon_i$. Let $n\in\N$ be such that $\frac{1}{c_3}< n\tau$ and define $t^\star_i = t_i + \frac{1}{n}\left(\frac{1-\epsilon_i}{c_3}\right)$, noting $t_i^\star<t_i+\tau$. It follows from the Taylor expansion that:
\begin{equation*}
    \Vert\mb{e}(t_i^\star)\Vert_2^2 \geq \frac{(1-\epsilon_i)^3}{c_3^2}\frac{n-1}{n^3}.
\end{equation*}
As $\epsilon_i$ can be chosen arbitrarily small, we choose it such that:
\begin{equation*}
    (1-\epsilon_i)^3 \geq \frac{\sigma^2n^3}{4r^6(n-1)}\epsilon_i^2,
\end{equation*}
which indicates that:
\begin{equation*}
    2r^3\Vert\mb{e}(t^\star_i)\Vert_2 \geq \sigma\vert h(\mb{x}(t_i))\vert \geq \sigma\vert h(\mb{x}(t^\star))\vert, 
\end{equation*}
as $h(\mb{x}(t))$ is monotonically decreasing. As $t_i^\star<t_i+\tau$ this contradicts that $\tau$ is the MIET. The number of events as a function time and distance from the barrier are seen in Figure \ref{fig:cex}. Simulated interevent times for this system are shown in Figure \ref{fig:simresults} (see the blue curves). \end{proof}

\begin{figure}[b]
    \centering
    \hspace*{-.5cm}
    \includegraphics[trim = {0, 0, 0, 0 cm}, clip ,scale =0.4, valign=t]{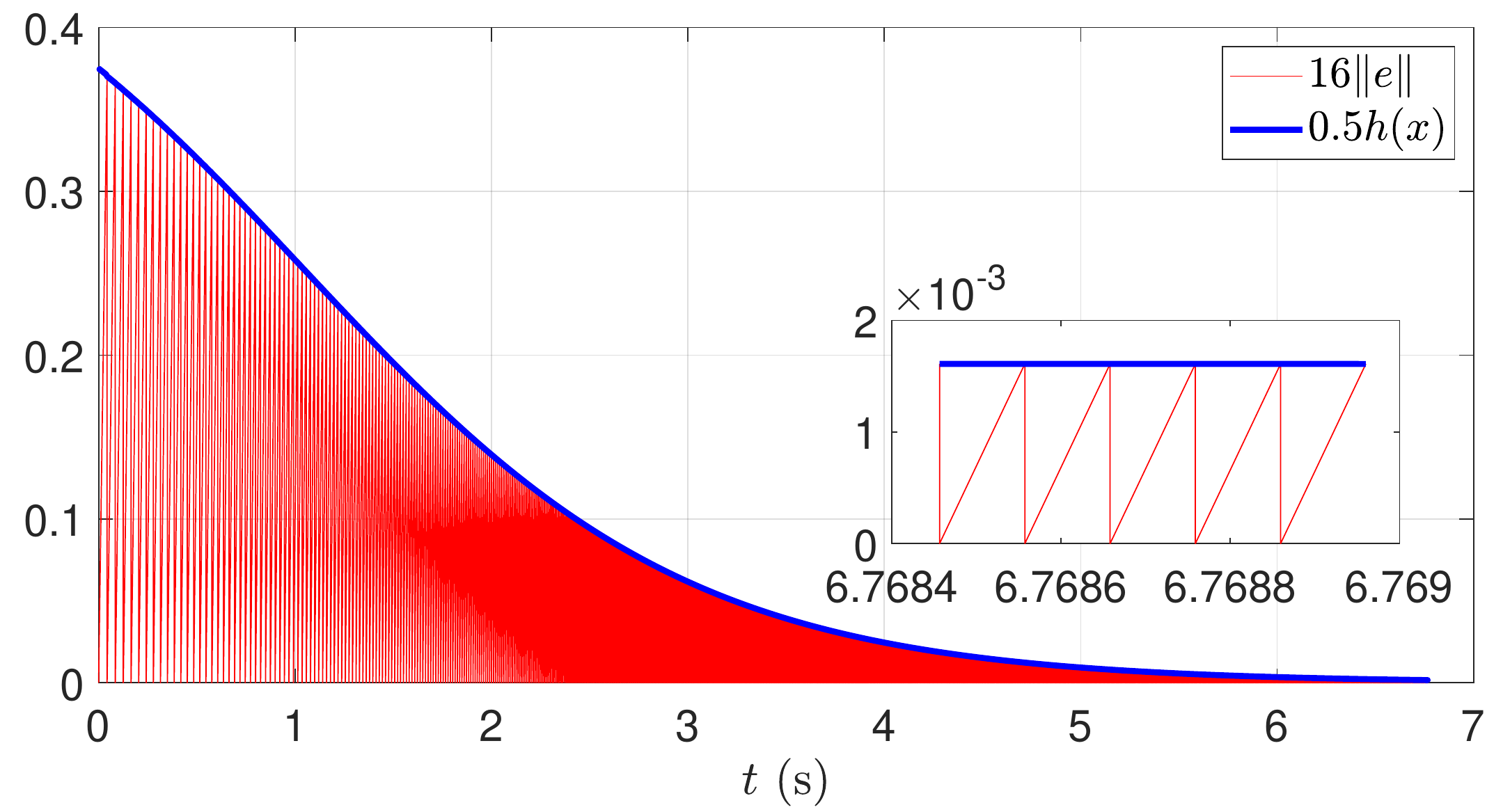}
    \caption{Simulation results for the system \eqref{eqn:cexcloop} using the trigger law \eqref{eqn:cextrigger}. Even as the boundary of the safe set is approached $h(t)\to 0$, the growth rate of the error does not diminish, leading to arbitrarily small interevent times.}
    \label{fig:cex}
\end{figure}

Throughout the proof of Lemma~\ref{le:no-MIET}, a critical concept arises at multiple steps. The fact that the state dynamics \eqref{eqn:cexcloop} are not required to converge to $\mb{0}$ at the boundary of the safe set $\C$ leads to the derivative of the measurement error, $\frac{\mathrm{d}}{\mathrm{d}t}\Vert\mb{e}(t_i)\Vert_2$, being uniformly lower bounded at event times $t_i$, which, together with the convergence of $h$ to $0$, cf. Figure \ref{fig:cex}, leads in turn to arbitrarily small interevent times.
%
%
In particular, the dynamics are allowed to evolve tangentially to the boundary of the safe set, which leads to growing measurement error while moving arbitrarily close to the 0-level set of $h$. 

As the original controller may have additional objectives beyond safety (such as stabilization), it is desirable that the event-triggered implementation not completely eliminate tangential motion near the boundary that may be necessary to achieve the other objectives. To accommodate this, we will introduce a trigger law that limits dynamic evolution tangential to the boundary of the safe set.

\section{Event-Triggered Safety}
\label{sec:etsafety}
In this section, we propose an alternative trigger to the one formulated in Section~\ref{sec:lyapbar} that ensures a MIET exists. To resolve the issues highlighted in the preceding counterexample, we introduce the following definition:

%
%

\begin{definition}[\textit{Strong ISSf Barrier Property}]
An ISSf-BF $h$ satisfies the \textit{strong ISSf barrier property} if there exists $d\in\R$ with $d>0$ such that for all $\mb{x},\mb{e}\in\R^n$:
\begin{equation}
\label{eqn:sbp}
        \derp{h}{\mb{x}}(\mb{x})\mb{f}(\mb{x},\mb{k}(\mb{x}+\mb{e})) \geq -\alpha(h(\mb{x}))+d-\iota(\Vert\mb{e}\Vert_2),
\end{equation}
\end{definition}

This property introduces a positive constant, $d$, into the ISSf-BF condition \eqref{eqn:ISSf-BF}. In the presence of zero measurement error, this enforces that the state dynamics must lie in the interior of the tangent cone \cite{blanchini2008set} when on the boundary of the safe set $\C$. It also enforces that $\frac{\mathrm{d}}{\mathrm{d}t}\vert\sigma h(\mb{x}(t_i))\vert$ will be greater than a positive constant as we approach the boundary, similarly to $\frac{\mathrm{d}}{\mathrm{d}t}\Vert\mb{e}(t_i)\Vert_2^2$. We now show this property is sufficient to design a trigger law that ensures safety with a MIET. 


\begin{theorem}[Trigger Law for Safety Critical Systems]
\label{thm:safetrig}
Let $h$ be an ISSf-BF for \eqref{eqn:clooperr} on a set $\C\subset\R^n$ defined as in \eqref{eqn:safeset}-\eqref{eqn:safetsetinterior}, with corresponding functions $\alpha\in\K_{\infty,e}$ and $\iota\in\K_\infty$. Let $\beta\in\K_{\infty,e}$ and $\sigma \in (0,1]$. If the following assumptions are satisfied:
\begin{enumerate}
    \item $h$ satisfies the strong ISSf barrier property for a constant $d\in\R$, $d>0$,
    \item $\iota$ is Lipschitz continuous with Lipschitz constant $L_\iota$,
    \item There exists $F\in\R$, $F>0$, such that for all $\mb{x},\mb{e}\in\R^n$:
    \begin{equation*}\Vert\mb{f}(\mb{x},\mb{k}(\mb{x}+\mb{e}))\Vert_2\leq F,
    \end{equation*}
    \item $\beta(r)\geq\alpha(r)$ for all $r\in\R$,
\end{enumerate}
then the trigger law:
\begin{align}
\label{eqn:correctrigger}
   t_{i+1} = \min\big\{t\geq t_i ~|~ \iota(\Vert\mb{e}(t)\Vert_2) &= \beta(h(\mb{x}(t))) 
   \\
   & \quad -\alpha(h(\mb{x}(t))) + \sigma d\big\}, \notag 
\end{align}
deployed recursively enforces:
\begin{equation}
 \dot{h}(\mb{x},\mb{e}) =  \derp{h}{\mb{x}}(\mb{x})\mb{f}(\mb{x},\mb{k}(\mb{x}+\mb{e})) \geq -\beta(h(\mb{x})),
\end{equation}
thus rendering the set $\C$ safe and asymptotically stable. Furthermore, there exists a MIET given by:
\begin{equation*}
    t_{i+1}-t_i \geq \tau \triangleq \frac{\sigma d}{L_{\iota}F},\quad i\in\N.
\end{equation*}
\end{theorem}

Before proving the result, we make a few observations regarding its assumptions. Assumption 3 on the boundedness of the dynamics need not hold over the entire state space for safety, but can hold for $(\mb{x},\mb{e})\in\C\times\R^n$. Furthermore, if $\C$ is compact, the trigger law enforces the existence of a compact set $E\subset\R^n$ such that $\mb{e}\in E$. Thus, the continuity of $\mb{f}$ and $\mb{k}$ implies the assumption is satisfied on $\C\times E$. This property needs to hold over a set containing $\C\times\R^n$ if asymptotic stability of $\C$ is desired.
We note that a similar bound on the dynamics is effectively required in the event-triggered stabilization formulation in \cite{tabuada2007event}, where it arises via a Lipschitz argument.

Assumption 4 ensures the right-hand side of the equality in the trigger will always be positive. $\beta$ can be thought of as a tuning function, which can practically raise interevent times (but not the MIET) at the expense of less ``braking" within the safe set and convergence outside of the safe set. One choice is $\beta=\alpha$, in which case interevent times are lowered for more braking and faster convergence.
\begin{proof}
To see the set $\C$ is rendered safe, observe that: \begin{align*}
      \dot{h} =  \derp{h}{\mb{x}}(\mb{x})\mb{f}(\mb{x},\mb{k}(\mb{x}+\mb{e})) \geq& -\alpha(h(\mb{x}))+d-\iota(\Vert\mb{e}\Vert_2),\\
        \geq& -\alpha(h(\mb{x}))+d ,\\
       & - (\beta(h(\mb{x}))-\alpha(h(\mb{x}))+\sigma d), \\
        = & -\beta(h(\mb{x})) + (1-\sigma)d,\\
        \geq & -\beta(h(\mb{x})),
\end{align*}
implying the barrier condition \eqref{eqn:barrier} is satisfied. Thus the set $\C$ is rendered safe.
%
%

To see the interevent time is lower bounded, observe that
\begin{align*}
    \Vert\mb{e}(t)\Vert_2 = & \left\Vert\mb{e}(0) +\int_{t_i}^t(-\mb{f}(\mb{x}(\tau),\mb{e}(\tau)))\mathrm{d}\tau\right\Vert_2 \\
    = & \left\Vert \int_{t_i}^t(-\mb{f}(\mb{x}(\tau),\mb{e}(\tau)))\mathrm{d}\tau\right\Vert_2 
    \leq \int_{t_i}^tF\mathrm{d}\tau .
\end{align*}
This inequality together with the trigger law~\eqref{eqn:correctrigger} yields:
\begin{align*}
\label{eqn:correctrigger}
   t_{i+1} 
   \geq & \min\left\{t\geq t_i ~|~ L_\iota\Vert\mb{e}(t)\Vert_2 = \sigma d\right\}, \\
   \geq & \min\left\{t\geq t_i ~|~ L_\iota F(t-t_i)=\sigma d\right\}
   =  \frac{\sigma d}{L_\iota F}+t_i,
\end{align*}
ensuring the desired result. This in conjunction with the barrier function condition implies $\C$ is asymptotically stable.
\end{proof}

\begin{figure*}[ht]
    \hspace*{-0.5 cm}
     \centering
    \begin{subfloat}
        {\includegraphics[trim = {0, 0, 0, 0 cm}, clip ,scale =0.4, valign=t]{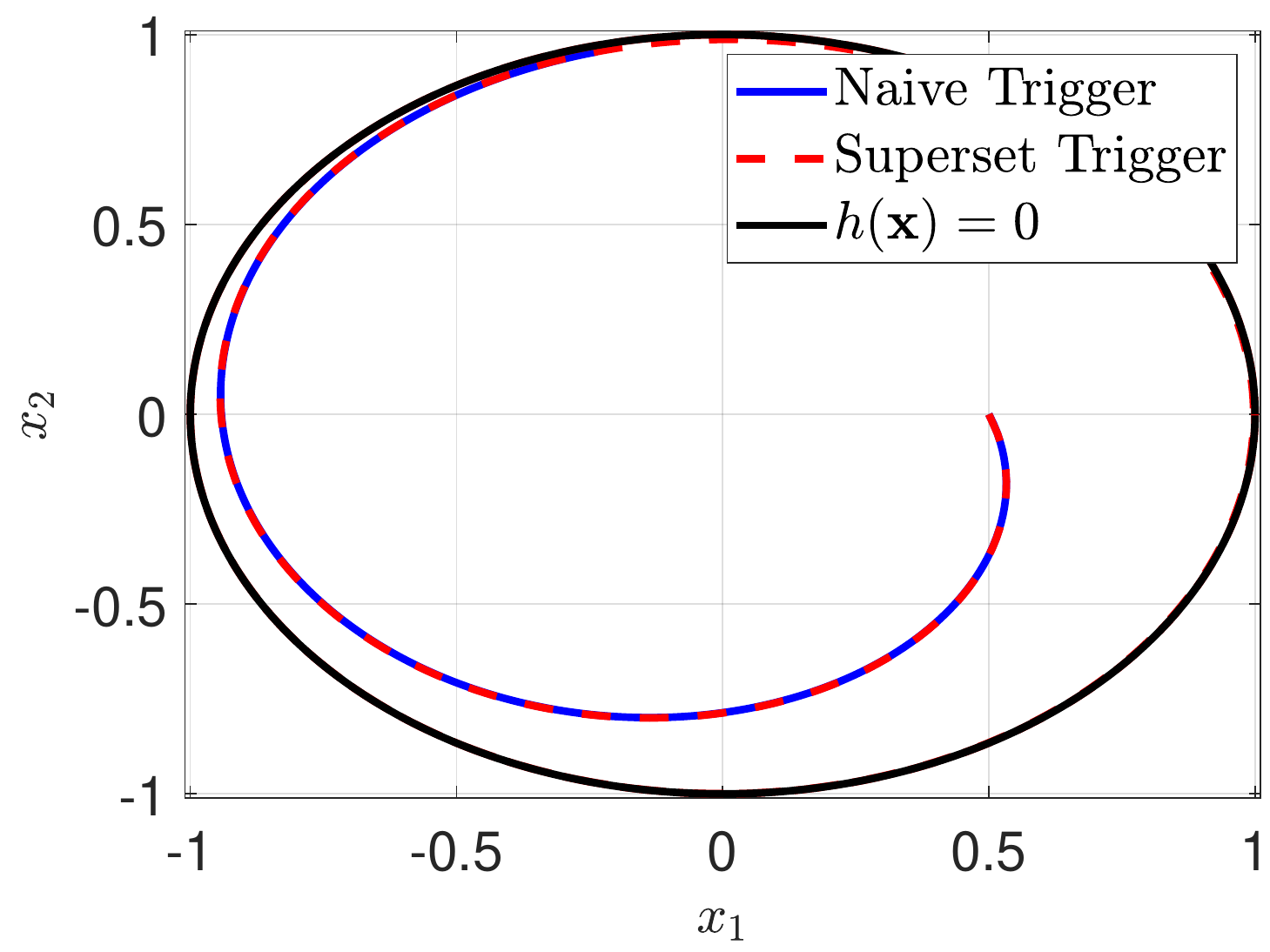}}
    \end{subfloat}
    \hfill
    \hspace*{0 cm}
    \begin{subfloat}
        {\includegraphics[scale =0.4, valign =t]{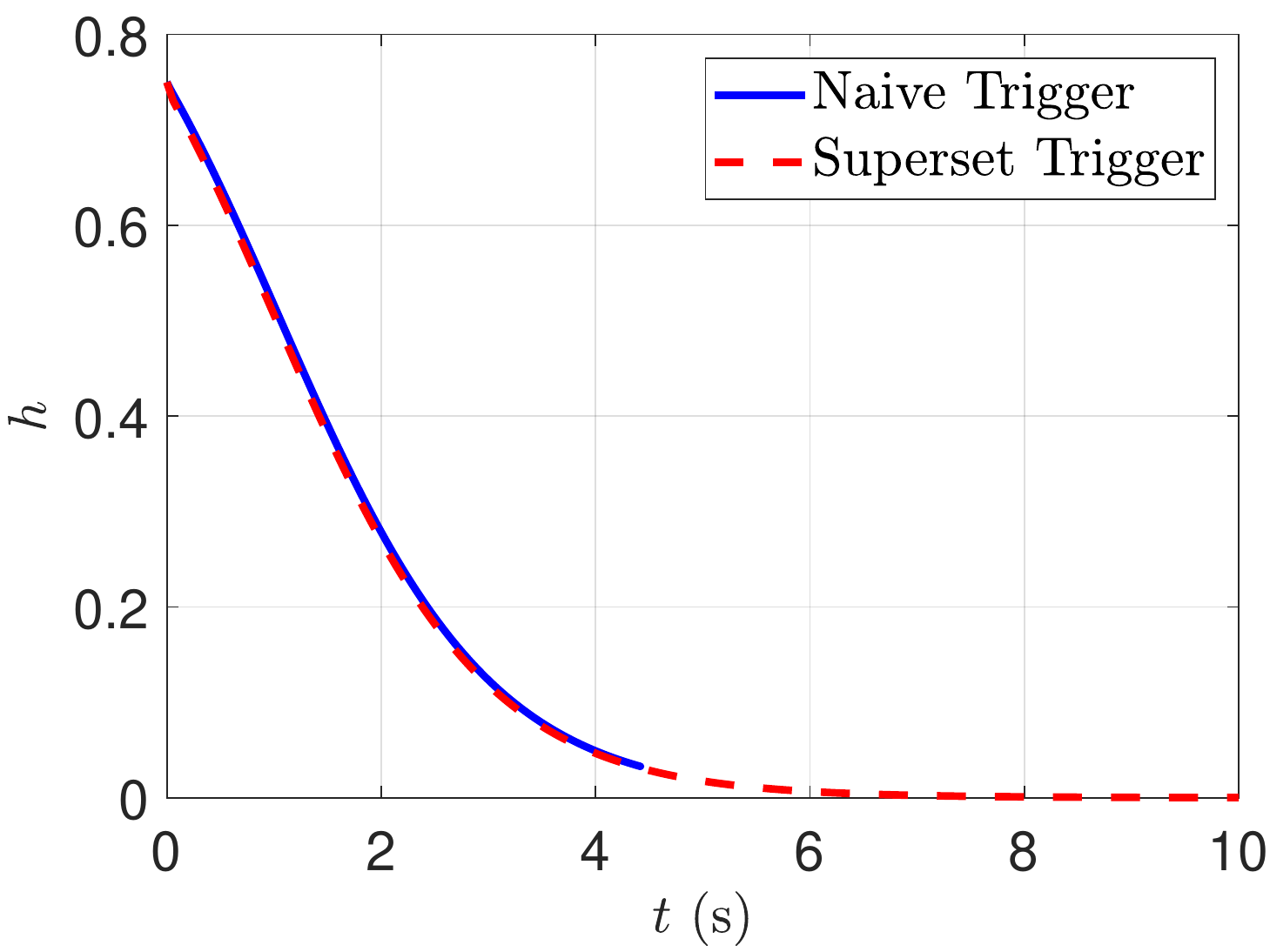}}
    \end{subfloat}
    \hfill
    \begin{subfloat}
      {\includegraphics[scale =0.4,valign=t]{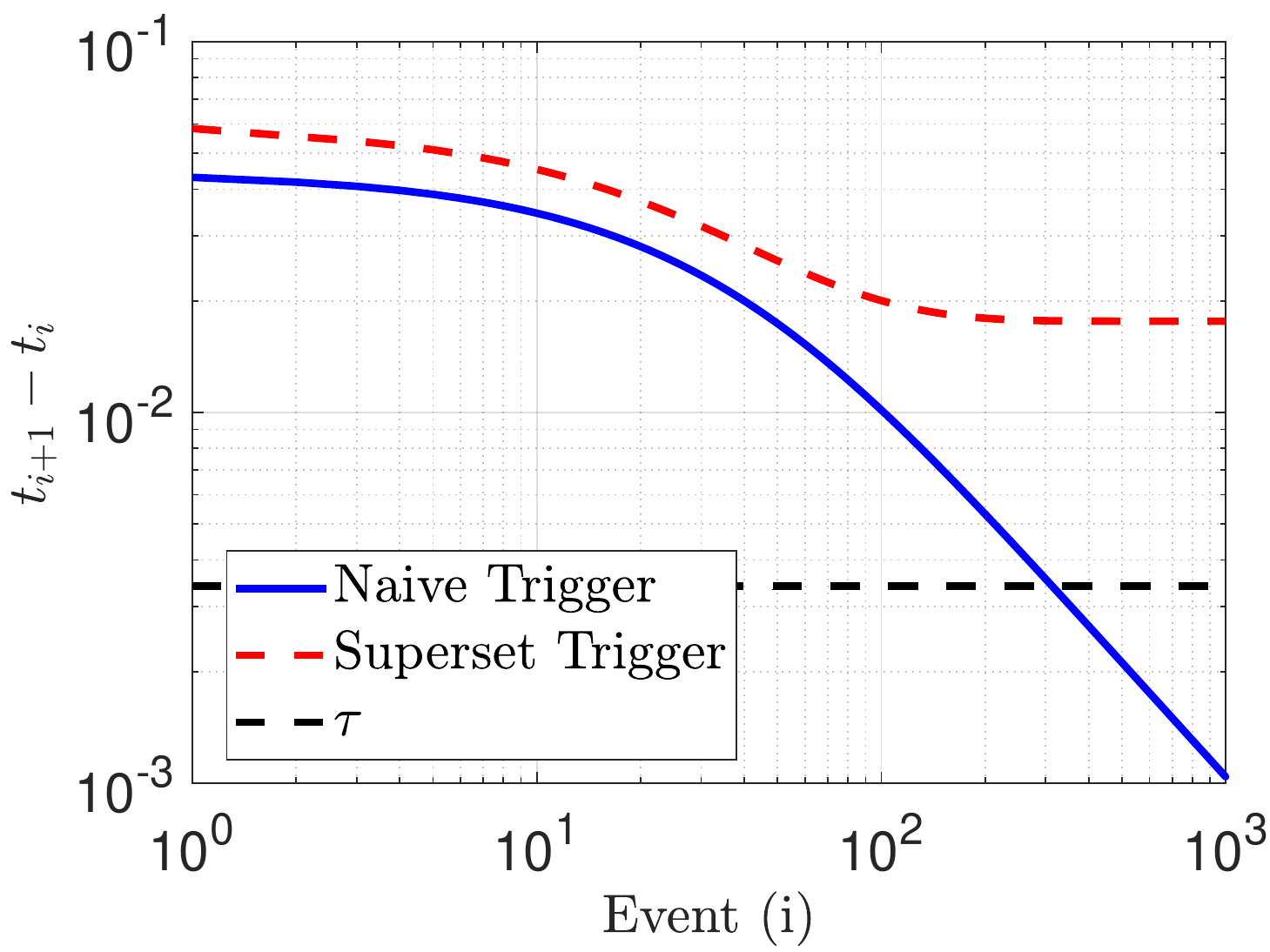}}
    \end{subfloat}
        \caption{Simulation results demonstrating safety achieved with an event-triggered controller. \textbf{(Left)} State trajectories for both triggers remain within the safe set for the length of the simulation. \textbf{(Center)} The value of the ISSf-BF $h$ remains above zero for the length of the simulation, corresponding to the system remaining safe.
        \textbf{(Right)} The interevent times of the two trigger laws. The interevent times of the trigger law \eqref{eqn:cextrigger} approach $0$ while the trigger law \eqref{eqn:correctrigger} satisfies the theoretical bound.
        }
    \label{fig:simresults}
\end{figure*}


In the case that an ISSf-BF $h$ does not satisfy the strong barrier property, an auxiliary ISSf-BF, $h_b$, satisfying the strong ISSf barrier property can be synthesized via $h$ at the expense of guaranteeing only a larger set is kept safe. This is formalized in the following result.

\begin{theorem}[Strong ISSf Barrier Property in Supersets]
Let $h$ be an ISSf-BF for \eqref{eqn:clooperr} on a set $\C\subset\R^n$ defined as in \eqref{eqn:safeset}-\eqref{eqn:safetsetinterior}, with corresponding functions $\alpha\in\K_{\infty,e}$ and $\iota\in\K_{\infty}$. Then the function $h_b$ defined as $h_b(\mb{x}) = h(\mb{x})+b$, with $b\in\R$, $b>0$, is an ISSf-BF satisfying the strong ISSf barrier property on the set $\C_b$ defined as:
\begin{equation}
    \C_b \triangleq \{\mb{x}\in\R^n ~|~ h_b(\mb{x})\geq 0\}
\end{equation}
\end{theorem}
\begin{proof}
Observe that:
\begin{align*}
&
    \derp{h_b}{\mb{x}}(\mb{x})\mb{f}(\mb{x},\mb{k}(\mb{x}+\mb{e})) =  \derp{h}{\mb{x}}(\mb{x})\mb{f}(\mb{x},\mb{k}(\mb{x}+\mb{e})) \\
    & \quad \geq  -\alpha(h(\mb{x})) - \iota(\Vert\mb{e}\Vert_2) 
    \geq -\alpha(h_b(\mb{x})-b)+\alpha(-b) \\
    & \qquad -\alpha(-b) - \iota(\Vert\mb{e}\Vert_2)
    =  -\alpha_b(h_b(\mb{x}))+d_b-\iota(\Vert\mb{e}\Vert_2),
\end{align*}
where $\alpha_b\in\cal{K}_{\infty,e}$ is defined as $\alpha_b(r) = \alpha(r-b)-\alpha(-b)$ and $d_b=-\alpha(-b)>0$.
\end{proof}

This result enables us to use an ISSf-BF with the trigger law in Theorem~\ref{thm:safetrig} by increasing the safe set by an arbitrarily small amount:

\begin{corollary}[Superset Trigger Law]
If $h$ is an ISSf-BF for \eqref{eqn:clooperr} on the set $\C$ satisfying Assumptions (2-4) of Theorem~\ref{thm:safetrig}, then $h_b$ is an ISSf-BF for \eqref{eqn:clooperr} on the set $\C_{b}$ satisfying Assumptions (1-4) of Theorem~\ref{thm:safetrig} such that the corresponding trigger law renders $\C_b$ safe and asymptotically stable with a MIET.
\end{corollary}

This is effectively an instance of Input-to-State Safety, in which case the original safe set $\C$ defined via $h$ becomes an ISSf safe set. We note that the larger the set is made (via a larger choice of $b$), the larger the MIET will be. This effectively highlights a trade-off that arises in the context of safety but not in stabilization: allowing motion near the boundary of a safe set requires additional relaxations to achieve the additional desirable property of the MIET.

To verify the ability of this trigger to keep the system safe and have a MIET, we simulated the system~\eqref{eqn:cex} using both the trigger law~\eqref{eqn:cextrigger} and the trigger law~\eqref{eqn:correctrigger}. The results of these simulations can be seen in Figure~\ref{fig:simresults}. We see that although both systems are kept safe, the trigger law not using the strong ISSf barrier property has interevent times that approach $0$.

\section{Conclusions}
We have presented a novel approach for achieving the safety of a system with a resource efficient event-triggered control law using Input-to-State Safe Barrier Functions. Similarities and differences between achieving stability and safety in an event-triggered context were highlighted through a counterexample, with a particular focus on how the behavior guaranteed with ISSf-BFs can lead to interevent times that are not lower bounded. This insight is used to propose a trigger law that renders the system input-to-state safe and guarantees a MIET for the system. Future work will include exploring adaptively choosing $b$ to improve interevent times, the synthesis of ISSf-BFs satisfying the strong input-to-state safe barrier property, and event-triggered co-design for stability and safety.

\bibliographystyle{plain}
\bibliography{taylor_main}
\end{document}